\pdfoutput=1
\PassOptionsToPackage{pdftex,
pdfencoding=auto,
pdfnewwindow=true,
pdfusetitle=true,
bookmarks=true,
bookmarksnumbered=true,
bookmarksopen=false,
pdfpagemode=UseThumbs,
bookmarksopenlevel=1,
pdfpagelabels=true,
breaklinks=true,
colorlinks=true,
unicode=true,
pdfborder={0 0 0},
backref=false
}{hyperref}
\PassOptionsToPackage{dvipsnames,pdftex}{xcolor}
\documentclass[superscriptaddress,aps,prl,nofootinbib,twocolumn,twoside,reprint,letterpaper,showpacs,10pt,floatfix,tikz,amsmath,amsfonts,amssymb,longbibliography]{revtex4-2}

\AtBeginDocument{
    \newwrite\bibnotes
    \def\bibnotesext{Notes.bib}
    \immediate\openout\bibnotes=\jobname\bibnotesext
    \immediate\write\bibnotes{@CONTROL{REVTEX42Control}}
    \immediate\write\bibnotes{@CONTROL{
    apsrev42Control,author="08",editor="1",pages="1",title="0",year="1"}}
     \if@filesw
     \immediate\write\@auxout{\string\citation{apsrev42Control}}
    \fi
}

\usepackage{amsthm,mathtools,bm,graphicx,verbatim,mathrsfs,units}
\usepackage[utf8]{inputenx}
\usepackage{bbm}
\usepackage{xcolor}
\usepackage{dsfont}
\usepackage{float}
\usepackage{cancel}
\usepackage{booktabs}
\usepackage{mathrsfs}
\usepackage[alwaysadjust,flushleft,pointlessenum]{paralist}
\setdefaultenum{\bfseries1.}{}{}{}

\usepackage[nolist,withpage]{acronym}
\usepackage{makecell}

\definecolor{myurlcolor}{rgb}{0,0,0.7}
\definecolor{myrefcolor}{rgb}{0.8,0,0}
\definecolor{purple}{RGB}{128,0,128}
\definecolor{ultramarine}{RGB}{63, 0, 255}
\definecolor{medblue}{RGB}{0, 0, 100}
\definecolor{googleblue}{RGB}{34, 0, 204}
\definecolor{panblue}{RGB}{0,24,150}
\definecolor{carmine}{RGB}{150, 0, 24}
\definecolor{gray}{RGB}{150, 150, 150}
\usepackage{hyperref} 
\hypersetup{
unicode=true,
bookmarksnumbered=false,
bookmarksopen=false,
breaklinks=false,
colorlinks=true,
linkcolor=carmine,
citecolor=googleblue,
urlcolor=panblue,
anchorcolor=OliveGreen}
\usepackage{amsthm}
\usepackage{amsmath}
\usepackage{amssymb}
\usepackage{graphicx}
\usepackage{color}
\usepackage{enumitem} 
\usepackage[caption=false]{subfig}
\usepackage{tikz}
\usetikzlibrary{shapes.geometric,shapes,positioning,shapes.misc}
\usetikzlibrary{shapes.misc,shapes.geometric,arrows.meta,calc}
\usepackage{amsmath}
\usepackage{pifont}
\usepackage{tabularx}
\usepackage{array}
\newtheorem*{postulate*}{Postulate}
\newtheorem*{principle*}{Principle}
\newtheorem*{theorem*}{Theorem}

\newcommand{\W}{\mathrm{\W}}
\newcommand{\GHZ}{\mathrm{\GHZ}}

\usepackage{stmaryrd} %

\newcommand{\beq}{\begin{equation}}
\newcommand{\eeq}{\end{equation}}
\newcommand{\beqa}{\begin{eqnarray}}
\newcommand{\eeqa}{\end{eqnarray}}
\newcommand{\bra}[1]{\ensuremath{\left\langle#1\right|}}
\newcommand{\ket}[1]{\ensuremath{\left|#1\right\rangle}}

\renewcommand{\today}{\number\day\space\ifcase\month\or
   January\or February\or March\or April\or May\or June\or
   July\or August\or September\or October\or November\or December\fi
   \space\number\year}

\definecolor{SkyBlue}{RGB}{135,206,235}

\usepackage{cprotect} %

\newtheorem{thm}{Theorem}
\newtheorem{theorem}[thm]{Theorem}

\newtheoremstyle{defblock}{0.7\topsep}{0pt}{}{}{}{: }{0pt plus 1pt minus 1pt}{\thmname{\bfseries{#1}}\thmnumber{\bfseries{#2}}\color{medblue}\bfseries\thmnote{#3}}
\theoremstyle{defblock}

\theoremstyle{remark}

\newcommand{\tagprop}[1]{\tag{\hyperref[#1]{P\ref{#1}}}}

\usepackage[normalem]{ulem}
\usepackage{microtype}
\microtypecontext{spacing=nonfrench}
\microtypesetup{
expansion={true,nocompatibility},
protrusion={true,nocompatibility},
activate={true,nocompatibility},
tracking=true,
kerning=true,
spacing={true}
}
\usepackage[all=normal,floats=tight,mathspacing=tight,wordspacing=tight,paragraphs=normal,tracking=tight,charwidths=tight,mathdisplays=normal,sections=normal,margins=normal]{savetrees}
\everypar=\expandafter{\the\everypar\loosness=-1 }
\usepackage[style=american,autopunct=true]{csquotes}

\usepackage[scr=boondoxupr]{mathalfa} %

\usepackage{hyperref}
\hypersetup{
linkcolor=carmine,
citecolor=googleblue,
urlcolor=panblue,
anchorcolor=OliveGreen}

\usepackage{placeins}%

\newcommand{\TripText}[1]{%
  \rotatebox{90}{\raisebox{0.2ex}{$#1$}}%
}

\begin{document}

\begin{abstract}
We propose a novel causal principle that is a genuinely multipartite extension of Reichenbach’s common cause principle, namely, the coordination principle: parties in a network can achieve perfect randomized coordination—in particular, agree
on a uniformly random output—only if they all share a common cause. We show that this principle does not follow from the standard no-signaling and independence principles by 
providing an explicit theory satisfying all these principles while violating the coordination principle. Strikingly, we prove that
the coordination principle holds, however, in quantum theory for four parties, and derive noise-tolerant Bell-like inequalities that certify a common cause. 
We then extend these results to a genuinely quantum coordination task, showing that the four-partite GHZ state requires a quantum common cause which
can also be certified by experimentally accessible Bell-like inequalities. A companion paper Ref.~\cite{CompanionPaperPRA} generalizes these results for N parties, proving that the coordination principle is satisfied in general for quantum theory.

\end{abstract}

\title{Coordination Requires a Common Cause in Quantum Theory}

\author{Daniel Centeno$^{\dagger,*}$}
\affiliation{Perimeter Institute for Theoretical Physics, Waterloo, Ontario, Canada.}
\affiliation{Department of Physics and Astronomy, University of Waterloo, Waterloo, Ontario, Canada, N2L 3G1}

\author{Antoine Coquet$^{\dagger,\star}$}
\affiliation{Inria, CPHT, LIX, CNRS, École polytechnique, Institut Polytechnique de Paris, Palaiseau, France}

\author{Maria~Ciudad~Alañón}
\affiliation{Perimeter Institute for Theoretical Physics, Waterloo, Ontario, Canada.}
\affiliation{Department of Physics and Astronomy, University of Waterloo, Waterloo, Ontario, Canada, N2L 3G1}

\author{Lucas Tendick}
\affiliation{Inria, CPHT, LIX, CNRS,  École polytechnique, Institut Polytechnique de Paris, Palaiseau, France}

\author{Marc-Olivier Renou}
\affiliation{Inria, CPHT, LIX, CNRS, École polytechnique, Institut Polytechnique de Paris, Palaiseau, France}

\author{Elie Wolfe}
\affiliation{Perimeter Institute for Theoretical Physics, Waterloo, Ontario, Canada.}
\affiliation{Department of Physics and Astronomy, University of Waterloo, Waterloo, Ontario, Canada, N2L 3G1}

\begingroup
\renewcommand{\thefootnote}{$\dagger$}
\footnotetext[0]{These authors contributed specially significantly.}
\endgroup

\begingroup
\renewcommand{\thefootnote}{$*$}
\footnotetext[0]{dcentenodiaz@perimeterinstitute.ca}
\endgroup

\begingroup
\renewcommand{\thefootnote}{$\star$}
\footnotetext[0]{antoine.coquet@inria.fr}
\endgroup

\maketitle

\emph{Introduction.---}
Bell nonlocality is one of the most striking features of quantum theory, with profound theoretical and philosophical implications~\cite{bell1964einstein, Brunner2014BellReview}. %
It shows that the observed correlations between distant systems cannot always be explained by classical reasoning, challenging our everyday intuition~\cite{Einstein1935EPR}.
Besides many practical consequences \cite{Supic2020selftestingof,Ekert1991Crypto,Acin2007DIQKD, Colbeck2009Randomness, Pironio2010Randomness}, this raises a fundamental question~\cite{Wheeler1989WhyQuantum, Chiribella2016QuantumTheory}: \emph{if classical reasoning fails, what fundamental principles determine which correlations can be observed in nature?}\\ %
\indent Over the last decades, several such principles have been proposed~\cite{POPESCU1992Generic, Pawowski2009Information, Fritz2013Local, Brassard2006Limit, Linden2007Computation,Gisin2020NSI} with the goal of singling-out quantum theory as the only reasonable explanation for the correlations observed in nature. Although these principles are not sufficient for a full recovery of all quantum correlations~\cite{Navascues2015Almost}, they nevertheless capture some of its fundamental features and set limits on potential refinements of quantum theory. \\
\indent
Conceptually, these proposed principles are of three different flavours: information-theoretic, complexity-theoretic, and causality-based. 
The first two invoke intricate auxiliary assumptions \cite{Pawowski2009Information, Fritz2013Local,Brassard2006Limit,Linden2007Computation}. 
In contrast, existing causality principles (No-Signaling faster than light, and Independence of causally unrelated %
variables~\cite{Coiteux2021any, Coiteux2021no, Henson2014, Wolfe2016inflation, Gisin2020NSI, Reichenbach,popescu1994quantum}) appeal through their simplicity. %
Yet these simple causality principles impose only coarse constraints on admissible theories and thus have limited power to rule out unreasonable ones~\cite{Gisin2020NSI, popescu1994quantum}. %
This raises the question: \emph{Are No-Signaling and Independence the unique relevant causality principles, or do additional causal principles exist?}\\
\indent To answer this question, we consider the task of \emph{coordination}, i.e., the agreement on a random common output among a set of parties. 
We introduce a \emph{Coordination Principle}:  a (reasonable) causal theory of information should enable such coordination \emph{only when} all parties share a \emph{common cause}\footnote{Note that in this work we use the term ``cause'' to refer to direct and indirect causes. More formally, 
$X$ is a cause of a variable 
$Y$ if there exists a directed causal path from $X$ to $Y$ in the circuit under consideration. For example, in Fig.~\ref{fig:circuit}, the \emph{causes} of $A$ are the transformations $AB$, $AC$, $AD$ and the sources $ABC$, $ABD$ and $ACD$.
}. 
Although classical theories respect this principle \cite{steudel2015information}, it is not clear (i) whether quantum theory does, or, more generally, (ii) if this is a natural consequence of the simple No-Signaling and Independence causality principles. In this letter, we answer these questions. 
\begin{figure}[h]
\centering

\def\Y{0.57}    %
\def\M{0.53}   %

\scalebox{0.88}{ %
\begin{tikzpicture}[
  arrow/.style={->, >=Latex},
  trip/.style={
    draw,
    shape=semicircle,
    shape border rotate=90,
    minimum width=5mm, %
    minimum height=5mm,
    inner sep=1pt
  },
  single/.style={
    draw,
    shape=semicircle,
    shape border rotate=270,
    minimum width=4mm, %
    minimum height=4mm,
    inner sep=1pt
  },
  pair/.style={
    draw,
    shape=diamond,
    minimum width=10mm,
    minimum height=10mm,
    inner sep=1pt
  }
]

\node[trip] (ABC) at (0,   3*\Y) {\rotatebox{90}{$ABC$}};
\node[trip] (ABD) at (0,   1*\Y) {\rotatebox{90}{$ABD$}};
\node[trip] (ACD) at (0,  -1*\Y) {\rotatebox{90}{$ACD$}};
\node[trip] (BCD) at (0,  -3*\Y) {\rotatebox{90}{$BCD$}};

\node[pair] (AB) at (3,   5*\M) {$AB$};
\node[pair] (AC) at (3,   3*\M) {$AC$};
\node[pair] (AD) at (3,   1*\M) {$AD$};
\node[pair] (BC) at (3,  -1*\M) {$BC$};
\node[pair] (BD) at (3,  -3*\M) {$BD$};
\node[pair] (CD) at (3,  -5*\M) {$CD$};

\node[single] (A) at (6,   3*\Y) {$A$};
\node[single] (B) at (6,   1*\Y) {$B$};
\node[single] (C) at (6,  -1*\Y) {$C$};
\node[single] (D) at (6,  -3*\Y) {$D$};

\node[inner sep=0pt] (a) [right=3mm of A] {$a$};
\node[inner sep=0pt] (b) [right=3mm of B] {$b$};
\node[inner sep=0pt] (c) [right=3mm of C] {$c$};
\node[inner sep=0pt] (d) [right=3mm of D] {$d$};

\draw[arrow] (A) -- (a);
\draw[arrow] (B) -- (b);
\draw[arrow] (C) -- (c);
\draw[arrow] (D) -- (d);

\draw[arrow] (ABC) -- (AB);
\draw[arrow] (ABC) -- (AC);
\draw[arrow] (ABC) -- (BC);

\draw[arrow] (ABD) -- (AB);
\draw[arrow] (ABD) -- (AD);
\draw[arrow] (ABD) -- (BD);

\draw[arrow] (ACD) -- (AC);
\draw[arrow] (ACD) -- (AD);
\draw[arrow] (ACD) -- (CD);

\draw[arrow] (BCD) -- (BC);
\draw[arrow] (BCD) -- (BD);
\draw[arrow] (BCD) -- (CD);

\draw[arrow] (AB) -- (A);
\draw[arrow] (AB) -- (B);

\draw[arrow] (AC) -- (A);
\draw[arrow] (AC) -- (C);

\draw[arrow] (AD) -- (A);
\draw[arrow] (AD) -- (D);

\draw[arrow] (BC) -- (B);
\draw[arrow] (BC) -- (C);

\draw[arrow] (BD) -- (B);
\draw[arrow] (BD) -- (D);

\draw[arrow] (CD) -- (C);
\draw[arrow] (CD) -- (D);

\end{tikzpicture}
}
\caption{\textbf{Coordination task between $A$, $B$, $C$ and $D$ without a common cause.} Parties $A$, $B$, $C$ and $D$ each output a bit $a$, $b$, $c$ and $d$, respectively. Their task is to attempt to coordinate their outcome, i.e., obtain a uniformly random bit $a=b=c=d$.
The parties are assisted by shared information received according to the above circuit. %
Note that the circuit we describe is to be understood through a causal lens, i.e., any wire is represented by an arrow to indicate a causal influence. %
Additionally, we consider the systems carrying information in the wires to be of arbitrary dimension. The leftmost nodes ($ABC$, $ABD$, $ACD$, $BCD$) are \emph{sources of information}, the intermediate nodes ($AB, AC,\dots, CD$) are %
\emph{transformations of information}, and the party nodes ($A, B, C, D$) are \emph{measurements of information}.
\emph{Crucially} this circuit is such that the parties do not all share a \emph{common cause} (e.g., $ABC$ does not influence $D$), in other words they do not all receive information emanating from a common source in our circuit.  
We investigate whether coordination can be achieved in this circuit. First, we show that it is possible to construct an Operational Probabilistic Theory (OPT) of information where $A, B, C, D$ can successfully coordinate. Second, we prove that coordination is, however, not possible in quantum theory. We generalize these results to any number $N \geq 4$ of parties in~\cite{CompanionPaperPRA}.}
\label{fig:circuit}
\end{figure}
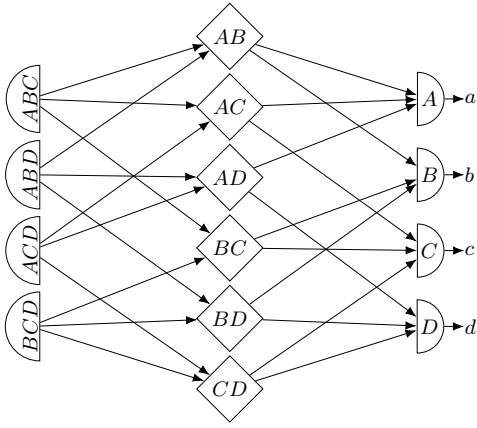\\
\indent We first provide a surprising answer to (ii) by showing that there do exist theories of information that obey the No-Signaling and Independence causality principles, but which violate our coordination principle. More precisely, we construct an Operational Probabilistic Theory (OPT) %
that satisfies the two former causal principles but allows for perfect coordination between parties not sharing a \emph{common cause} 
(see Fig.~\ref{fig:circuit}). \\
\indent Secondly, we provide an answer to (i) consistent with our intuition about reasonable causal theories by proving that nevertheless, within quantum theory, any perfect coordination must originate from a common cause. \\

\indent Then, we extend this second result to a purely quantum coordination task, namely, the creation of a multipartite GHZ state. We show that this quantum coordination task is impossible in the absence of a quantum common cause, \emph{even if} the parties are assisted by a global source of classical information. %
\\
\indent The last two results are derived assuming that the correlations in nature obey quantum theory. Under this assumption, we are able to provide noise-tolerant Bell-like inequalities that certify the existence of a common (quantum) cause and comment on their potential violation in recent quantum-optics experiments.
\\
\indent In contrast, our conceptual discussion concerns theories beyond quantum and is centered on the ideal notion of \emph{perfect} coordination. Since perfect coordination is a sharp, asymptotic property, its operational verification would require infinite statistics. We discuss the need for
an \emph{additional causality principle distinct from No-Signaling and Independence}, that would apply to any reasonable causal theory of information and would be \emph{able to quantitatively constrain} how parties lacking a common cause must fail to coordinate—not only perfectly, but also approximately in the presence of noise.

\indent \emph{Causality and coordination.---}
The notion of causality embodies our physical intuition regarding how variables can influence one another. Currently, it is formalized by the following two principles, \emph{No-Signaling} and \emph{Independence} (NSI)\cite{Gisin2020NSI}. 
\begin{principle*}[No-Signaling]
    A variable can only be influenced by its causal past. 
\end{principle*}
For instance, in the circuit\footnote{%
Note that the circuits we describe are to be understood through a causal lens, i.e., any wire is represented by an arrow to indicate a causal influence. Additionally, we consider that the wires contain arbitrary dimension systems carrying information.} of Fig.~\ref{fig:circuit}, the outcome of $A$ cannot be altered by any modification of the source $BCD$.
\begin{principle*}[Independence]
    Two variables that do not share a common cause cannot be correlated.\footnote{Note that we use the term variable to refer not only to atomic circuit nodes but also to sub-circuit formed through composition of elementary circuit nodes.}
\end{principle*}
This principle is the contrapositive of Reichenbach’s principle~\cite{Reichenbach}, which states that two \emph{correlated} variables must share a common cause. When extended to $N$ correlated variables (i.e., $N$ variables whose joint probability distribution does not factorize), it only guarantees that pairs of variables share a common cause. However,
it does not ensure the existence of a common cause shared by all $N$ variables, a limitation that has been repeatedly pointed out in the literature~\cite{henson2005comparing,uffink1999principle,steudel2015information,Henson2014}. %
To address this, %
we propose the following principle:

\begin{principle*}[Coordination]
$N$ perfectly coordinated variables must share a common cause.%
\end{principle*}
Note that this principle can also be phrased as its contrapositive (à la Independence): $N$ variables that do not share a common cause cannot be coordinated. Then, according to this principle, the parties $A,B,C,D$ in the circuit of Fig.~\ref{fig:circuit} should be unable to produce a coordinated random output as they lack a common cause. Conversely, if the variables share a common cause, then there exists a local (hidden) variable model (namely, one where the cause acts as shared randomness) that enables them to achieve perfect coordination.

\emph{A theory satisfying No-Signaling and Independence where coordination does not require a common cause---}
We now show that the Coordination principle is \emph{not} implied by the NSI Principles. %
This is surprising, as in many circuits our Coordination Principle \emph{can} be derived from the NSI principles~\cite{Coiteux2021any,Coiteux2021no}.%

More precisely, we construct an OPT~\cite{d2017quantum} that satisfies the NSI Principles yet allows four-partite coordination
without a shared common cause. 
An OPT is described by giving a set of sources, channels and measurements and assigning a probability distribution to every possible way of composing (wiring) them into a circuit.
An OPT satisfies the NSI Principles whenever these probability distributions obey the constraints specified by them—for example, the marginal of two variables without a common cause must factorize~\cite{d2017quantum}.

In~\cite{CompanionPaperPRA}, we present an explicit OPT involving four sources, six channels, and four measurements that satisfies all NSI constraints; yet, in the circuit of Fig.~\ref{fig:circuit}, yields a joint distribution in which the outputs $a,b,c,d$ are coordinated, thus realising a shared random bit.

Prior derivations of the Coordination Principle from the NSI principles in specific circuits relied on the \emph{non-fanout inflation technique}~\cite{Wolfe2016inflation}, which analyzes larger circuits in which some parts replicate the original circuit.
Here, these arguments are not working due to the existence of intermediate channels, not considered in these prior works~\cite{Henson2014,Coiteux2021any}. However, intermediate channels have been recently shown to be critical in order to produce some probability distributions \cite{centeno2024significance,wolfe2021quantum}. This reveals a fundamental limitation of NSI-based arguments in circuits where information can propagate through nontrivial channel layers.

\emph{Coordination requires a common cause in quantum theory---} %
The following theorem establishes that quantum theory obeys the Coordination Principle:
\begin{theorem}
    Quantum theory satisfies the Coordination Principle.
    \label{theoremPRL}
\end{theorem}
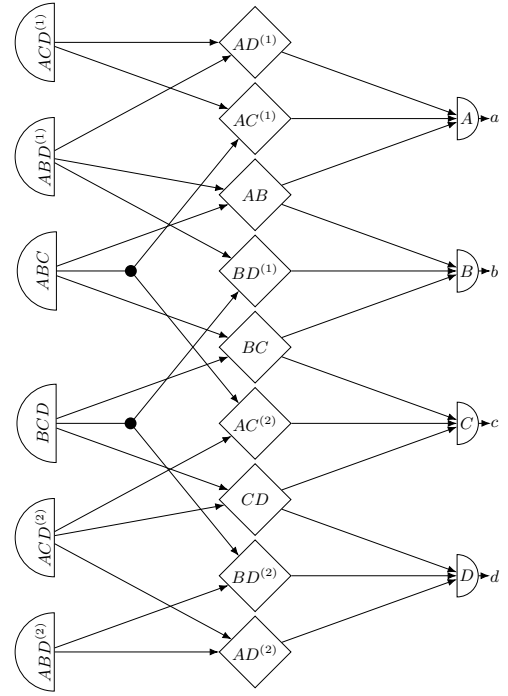
\begin{figure}[htbp]
\def\M{0.72}    %

\centering
\scalebox{0.7}{
\begin{tikzpicture}[
  arrow/.style={->, >=Latex},
  fadedarrow/.style={arrow, draw=black!40},
  trip/.style={
    draw,
    shape=semicircle,
    shape border rotate=90,
    minimum width=7.4mm,
    minimum height=7.4mm,
    inner sep=1pt,
    align=center
  },
  single/.style={
    draw,
    shape=semicircle,
    shape border rotate=270,
    minimum width=4mm,
    minimum height=4mm,
    inner sep=1pt
  },
  pair/.style={
    draw,
    shape=diamond,
    minimum width=13.6mm,
    minimum height=13.6mm,
    inner sep=1pt,
    align=center
  }
]

\node[pair]              (AD11) at (4,  8*\M) {$AD^{(1)}$};
\node[pair]              (AC11) at (4,  6*\M) {$AC^{(1)}$};
\node[pair]              (AB11) at (4,  4*\M) {$AB$};
\node[pair]  (BD11) at (4,  2*\M) {$BD^{(1)}$};  %
\node[pair]              (BC11) at (4,  0*\M) {$BC$};
\node[pair]  (AC12) at (4, -2*\M) {$AC^{(2)}$};  %
\node[pair]              (CD21) at (4, -4*\M) {$CD$};
\node[pair]              (BD21) at (4, -6*\M) {$BD^{(2)}$};
\node[pair]              (AD22) at (4, -8*\M) {$AD^{(2)}$};

\path (AC11) ++(4,0) node[single]             (A) {$A$};
\path (BD11) ++(4,0) node[single] (B) {$B$};  %
\path (AC12) ++(4,0) node[single] (C) {$C$};  %
\path (BD21) ++(4,0) node[single]             (D) {$D$};

\node[inner sep=0pt]             (a) at ($(A.east)+(3mm,0)$) {$a$};
\node[inner sep=0pt] (b) at ($(B.east)+(3mm,0)$) {$b$};
\node[inner sep=0pt] (c) at ($(C.east)+(3mm,0)$) {$c$};
\node[inner sep=0pt]             (d) at ($(D.east)+(3mm,0)$) {$d$};

\draw[arrow]      (A) -- (a);
\draw[arrow] (B) -- (b);
\draw[arrow] (C) -- (c);
\draw[arrow]      (D) -- (d);

\node[trip] (ACD1) at ($(AD11)+(-4,0)$) {\TripText{ACD^{(1)}}};

\node[trip] (ABC1) at ($(BD11)+(-4,0)$) {\TripText{ABC}};

\node[trip] (ABD1) at ($(ACD1)!0.5!(ABC1)$) {\TripText{ABD^{(1)}}};

\node[trip] (BCD1) at ($(AC12)+(-4,0)$) {\TripText{BCD}};

\node[trip] (ABD2) at ($(AD22)+(-4,0)$) {\TripText{ABD^{(2)}}};

\node[trip] (ACD2) at ($(ABD2)!0.5!(BCD1)$) {\TripText{ACD^{(2)}}};

\draw[arrow] (ACD1) -- (AD11);
\draw[arrow] (ABD1) -- (AD11);

\draw[arrow] (ACD1) -- (AC11);

\draw[arrow] (ABD1) -- (AB11);
\draw[arrow] (ABC1) -- (AB11);

\draw[arrow] (ABD1) -- (BD11);

\draw[arrow] (ABC1) -- (BC11);
\draw[arrow] (BCD1) -- (BC11);

\draw[arrow] (ACD2) -- (AC12);

\draw[arrow] (BCD1) -- (CD21);
\draw[arrow] (ACD2) -- (CD21);

\draw[arrow] (ABD2) -- (BD21);

\draw[arrow] (ACD2) -- (AD22);
\draw[arrow] (ABD2) -- (AD22);

\draw[arrow]      (AD11) -- (A);
\draw[arrow]      (AC11) -- (A);
\draw[arrow]      (AB11) -- (A);

\draw[arrow]      (AB11) -- (B);
\draw[arrow]      (BC11) -- (B);
\draw[arrow] (BD11) -- (B);   %

\draw[arrow] (AC12) -- (C);   %
\draw[arrow]      (BC11) -- (C);
\draw[arrow]      (CD21) -- (C);

\draw[arrow]      (AD22) -- (D);
\draw[arrow]      (BD21) -- (D);
\draw[arrow]      (CD21) -- (D);

\begin{scope}[shift={(ABC1.east)},scale=0.8,
              line cap=round,line join=round]
  \colorlet{wire}{black}
  \colorlet{contact}{black}
  \colorlet{contactinactive}{black!35}
  \colorlet{lever}{black}
  \colorlet{ghostlever}{black!40}
  \colorlet{arrowclr}{black!70}

  \coordinate (ABCin)    at (0,0);
  \coordinate (ABCpivot) at (1.75,0);

  \draw[wire] (ABCin) -- (ABCpivot);

  \fill[contact] (ABCpivot) circle (4pt);

  \draw[arrow]      (ABCpivot)   -- (AC11);   %
  \draw[arrow] (ABCpivot) -- (AC12);   %

\end{scope}

\begin{scope}[shift={(BCD1.east)},scale=0.8,
              line cap=round,line join=round]
  \colorlet{wire}{black}
  \colorlet{contact}{black}
  \colorlet{contactinactive}{black!35}
  \colorlet{lever}{black}
  \colorlet{ghostlever}{black!40}
  \colorlet{arrowclr}{black!70}

  \coordinate (BCDin)    at (0,0);
  \coordinate (BCDpivot) at (1.75,0);
  \coordinate (BCDup)    at (2,0.5);
  \coordinate (BCDdown)  at (2,-0.5);

  \draw[wire] (BCDin) -- (BCDpivot);

  \fill[contact] (BCDpivot) circle (4pt);

  \draw[arrow] (BCDpivot)   -- (BD11);  %
  \draw[arrow]      (BCDpivot) -- (BD21);%

\end{scope}

\end{tikzpicture}
}
\caption{
\textbf{Graphical representation of the quantum inflation used to demonstrate that Quantum theory satisfies the Coordination Principle.}
The superindex in parenthesis indicate the copy index whenever there is more than one copy. The black dots highlight whenever the same Hilbert space of a given source state is pertinent in defining more than one transformation. Accordingly, pairs of transformations related by a black dot ($AC^{(1)}$ and $AC^{(2)}$ or $BD^{(1)}$ and $BD^{(2)}$) cannot be physically implemented simultaneously. That implies incompatibility of their downstream measurements ($A$ and $C$ or $B$ and $D$ respectively). Nevertheless, while the full circuit is not physically valid, certain subcircuits remain well defined: in particular, any subcircuit obtained by following only one of the outgoing wires from each black dot is implementable. Some of these valid subcircuits reproduce the original circuit on the corresponding parties—for instance, the adjacent pairs- $(A,B),(B,C),(C,D)$. %
The proof outline is to first assume that perfect coordination is achievable in Fig.~\ref{fig:circuit}. Then, by consistency between the quantum inflation and the original circuit, the marginal distributions on the pairs $(A,B)$, $(B,C)$, $(C,D)$ must be that of a shared random bit. %
Using a sum of squares argument, the quantum formalism establishes that the pair $(A,D)$ must also output a shared random bit. However, $A$ and $D$ are causally independent in the quantum inflation, thus, according to the Independence principle, they must output uncorrelated bits which is in contradiction with outputing a shared random bit.}
\label{fig:inflation-circuit}
\end{figure}

\begin{proof}[Proof outline]
We illustrate the proof strategy with four events; a complete treatment of the general case is given in~\cite{CompanionPaperPRA}.
~\\
\emph{Step 1: Reduction to a canonical circuit.}
As shown in~\cite{CompanionPaperPRA}, any operational prediction in any circuit with measurement nodes $A,B,C,D$ that do not share a common cause can be achieved by the circuit of Fig.~\ref{fig:circuit}.  
To see this, one 
labels each source or channel by the measurement nodes in its causal future and merges nodes with identical labels: one obtains a circuit that can be embedded in Fig.~\ref{fig:circuit}. Thus, it suffices to analyse Fig.~\ref{fig:circuit}.
~\\
\emph{Step 2: Assume perfect coordination and inflate.} 
Suppose, toward a contradiction, that the circuit in Fig.~\ref{fig:circuit} can produce perfect coordination among $A,B,C,D$ within quantum theory. Then, there exist tripartite states $\ket{\psi}_{ABC},\ldots,\ket{\psi}_{BCD}$, unitaries $U_{AB},\ldots,U_{CD}$, and measurements $\Pi_A^a,\ldots,\Pi_D^d$ such that
\[
P(a,b,c,d)
= \langle \psi |\, U^{\dagger} \Pi^{a,b,c,d} U \,|\psi\rangle
=
\begin{cases}
1/2 & \text{if } a=b=c=d,\\[0.15em]
0 &\text{otherwise},
\end{cases}
\]
where $|\psi\rangle$, $U$, and $\Pi^{a,b,c,d}$ denote the tensor products of all states, unitaries, and projectors respectively. Then, we use the \emph{
quantum inflation technique} \cite{wolfe2021quantum}\footnote{The \emph{quantum inflation technique} also analyzes larger circuits in which some parts replicate the original circuit as the \emph{non-fanout inflation technique}. However, it also imposes the mathematical structure of Quantum Theory.}. In particular, we consider the specific quantum inflation of the original circuit (Fig.~\ref{fig:circuit}) which is graphically represented in %
Fig.~\ref{fig:inflation-circuit}. %
Note that the black dots used in the graphical representation of Fig.~\ref{fig:inflation-circuit} indicate whenever the same Hilbert space of a given source state is pertinent in defining more than one transformation. Accordingly, pairs of transformations related by a black dot (i.e., $AC^{(1)}$ and $AC^{(2)}$ or $BD^{(1)}$ and $BD^{(2)}$) cannot be physically implemented simultaneously. In turn, that implies incompatibility of their downstream measurements\footnote{The incompatibility of the measurements in the quantum inflation is mathematically implemented by removing commutation rules between the operators describing the measurements of those parties.} ($A$ and $C$ or $B$ and $D$ respectively). Nevertheless, while the full circuit is not physically valid, certain subcircuits remain well defined: in particular, any subcircuit obtained by following only one of the outgoing wires from each black dot is implementable, like the adjacent pairs, $(A,B),(B,C),(C,D)$, or the pair $(A,D)$.

\emph{Step 3: Consistency and contradiction.}
The proof establishes the following three facts:

(i) The two-party subnetworks of $(A,B)$, $(B,C)$ and $(C,D)$ are identical to the corresponding subnetworks of the original circuit. %
Thus, under our assumptions, the associated marginal distributions must be those of a \emph{shared random bit}. %
This follows directly from consistency between the behaviours in Fig.~\ref{fig:circuit} and Fig.~\ref{fig:inflation-circuit}. %

(ii) From (i), using the quantum formalism we can establish that the marginal distribution on $(A,D)$ must also be that of a shared random bit. To this end, we use the Heisenberg picture operators, $\widetilde{\Pi}_{A}^{a} \coloneqq U_{A}^{\dagger}\Pi_{A}^{a} U_{A}$, where $U_{A}$ is the tensor product of all the unitaries corresponding to the transformations in the past of $A$, namely  $U_{A}\coloneqq {U_{AD^{(1)}} \otimes U_{AC^{(1)}}\otimes U_{AB}}$. The remaining operators $\widetilde{\Pi}_{B}^{b}$, $\widetilde{\Pi}_{C}^{c}$ and $\widetilde{\Pi}_{D}^{d}$ are defined analogously. Now, since $(A,B)$ is jointly measurable in our inflation, a further consequence of the marginal being associated to a shared random bit is that 
\begin{align}
\langle \phi | \widetilde{\Pi}_{A}^{a}\widetilde{\Pi}_{B}^{b}  |\phi\rangle=\langle \phi | \widetilde{\Pi}_{A}^{a}  |\phi\rangle= \langle \phi | \widetilde{\Pi}_{B}^{b}  |\phi\rangle = \frac{1}{2}
\end{align}
where $|\phi\rangle$ is the tensor product of all the states used in the quantum inflation and hence
\begin{align}\begin{split}
||(\widetilde{\Pi}_{A}^{a}-\widetilde{\Pi}_{B}^{b})|\phi\rangle&||^2=\langle \phi |(\widetilde{\Pi}_{A}^{a}-\widetilde{\Pi}_{B}^{b})^2|\phi\rangle 
         \\&= \langle \phi | (\widetilde{\Pi}_{A}^{a})^2 -2\widetilde{\Pi}_{A}^{a}\widetilde{\Pi}_{B}^{b} + (\widetilde{\Pi}_{B}^{b})^2 |\phi\rangle = 0 \,.
\end{split}\end{align}
From this \emph{sum-of-squares argument} we conclude that $\widetilde{\Pi}_{A}^{0}\ket\phi = \widetilde{\Pi}_{B}^{0}\ket\phi$.
The same reasoning applies to the pairs of adjacent nodes $(B,C)$ and $(C,D)$. Thus, by transitivity, $\widetilde{\Pi}_{A}^{0}|\phi\rangle = \widetilde{\Pi}_{D}^{0}|\phi\rangle$.
Given that projectors are idempotent this would mean $\bra{\phi}\widetilde{\Pi}_{A}^{0} \widetilde{\Pi}_{D}^{0}\ket{\phi} = \bra{\phi}\widetilde{\Pi}_{A}^{0} \ket{\phi} = \bra{\phi}\widetilde{\Pi}_{D}^{0} \ket{\phi}$.

(iii) However, from the quantum inflation, we see
that $A$ and $D$ do not share any source in their past,
hence they are causally independent. Then, by the Independence principle, they must output \emph{independent} bits, i.e., $\bra{\phi}\widetilde{\Pi}_{A}^{0} \widetilde{\Pi}_{D}^{0}\ket{\phi} = \bra{\phi}\widetilde{\Pi}_{A}^{0} \ket{\phi}  \bra{\phi}\widetilde{\Pi}_{D}^{0} \ket{\phi} $.

The predictions (ii) and (iii) are incompatible: $(A,D)$ cannot be both perfectly coordinated and independent.  
Thus coordination in Fig.~\ref{fig:circuit} is impossible within quantum theory, completing the proof.
\end{proof}

The quantum inflation method is usually implemented computationally.
In~\cite{CompanionPaperPRA}, we obtain a noise-robust version of Theorem~\ref{theoremPRL} 
: within Quantum Theory, any four measurements producing bits and not sharing a common cause must satisfy:
\begin{equation}\label{eq:SharedRandBitIneq}
    \langle AB \rangle + \langle BC \rangle + \langle CD \rangle \leq \frac{\langle A\rangle \langle  D \rangle}{2} + \frac{3\sqrt{3}}{2} ,
\end{equation}
where the correlators are $\langle A \rangle=p(a=0)-p(a=1)$ and 
$\langle AB \rangle = p(a=b)-p(a\neq b)$, with analogous definitions for the other correlators.  
Inequality~\eqref{eq:SharedRandBitIneq} is a Bell-like inequality which, when violated, certifies that, either the four variables share a common cause, or 
the observed data cannot arise from a quantum process.  
For perfect shared random bits, the inequality reduces to $3\leq \frac{3\sqrt{3}}{2}\approx 2.598$, and is therefore maximally violated.

\emph{Obtaining a $\mathrm{GHZ}$ state requires a quantum common cause---}
We now move to a genuinely quantum version of the coordination task.  
Instead of producing classical bits, the four parties are required to output qubits whose joint state is the four-partite GHZ state 
$\ket{\mathrm{GHZ}}_{ABCD} = (\ket{0000}+\ket{1111})/\sqrt{2}$.  
Note that measuring this state in the computational basis yields the distribution of a perfect shared random bit, so it is clear that $\mathrm{GHZ}$ cannot be generated from the circuit of Fig.~\ref{fig:circuit}.

In~\cite{CompanionPaperPRA}, we also show that even if the four sources $ABC$, $ABD$, $ACD$, and $BCD$ share a \emph{classical} common cause, a GHZ state still cannot be produced in quantum theory.\footnote{
This already follows from the convex extremality of GHZ as a pure quantum state: a convex mixture of states yielding a GHZ state must consist solely of GHZ states. Thus, a pure GHZ could only be achieved \emph{with} shared randomess if it could already be achieved \emph{without} shared randomness, and the latter is already excluded per the impossibility of coordination without shared randomness.}
More precisely, when $A$, $C$ and $D$ each have two settings, and $B$ has three settings, if $p(C^1\cdot D^1=1)=p(C^1\cdot D^1 =-1)$ where the superindex indicates the setting, the following inequality holds:
\begin{equation}\label{eq:SharedGHZIneq}
\begin{split} \left(I_{CHSH}^{C^1\cdot D^1=-1}\circ \left\{ AB\right\}\right)^2  + \left(I_{CHSH}^{C^1\cdot D^1=1}\circ \left\{ AB\right\}\right)^2 \\ + 8(3[\langle A^0B^2\rangle + \langle B^2C^0\rangle + \langle C^0D^0\rangle]-8)^2\leq 16,
\end{split}
\end{equation}
where $I_{CHSH}^{C^1\cdot D^1=\pm 1}\circ \left\{ AB\right\}$ is the CHSH inequality conditioned on a particular output of $C$ and $D$ for their first settings.\footnote{Note that depending on the value on which we condition, we take a different version of the CHSH inequality, see the companion work~\cite{CompanionPaperPRA} for details.}

\emph{Experimental considerations---}
Consider first an experiment in which the parties $A,B,C,D$ output classical bits and violate inequality~\eqref{eq:SharedRandBitIneq}.  
Assuming quantum theory correctly describes the experiment, such a violation certifies that the four nodes share a common cause in their past.  
This is straightforward to implement in practice, as classical communication can distribute a shared random bit.

Now consider an experiment in which $A,B,C,D$ output qubits and violate inequality~\eqref{eq:SharedGHZIneq}.  
Under the assumption that quantum theory holds, such a violation certifies the presence of a \emph{quantum} common cause among the four nodes.  
For a noisy state  
$\rho = v\,\ket{\mathrm{GHZ}_4}\!\bra{\mathrm{GHZ}_4} + (1-v)\,\mathbb{I}/16$,  
our inequality is violated whenever $v>94.17\%$. %
Note that Ref.~\cite{CompanionPaperPRA} provides a general noise-tolerant inequality for the fidelity required to certify a quantum common cause for $N-$partite GHZ state correlations.

This threshold is compatible with fidelities reported in recent multipartite GHZ experiments~\cite{cao2022experimental}.  
However, existing four-photon GHZ experiments \emph{artificially simulate $\ket{\mathrm{GHZ}_4}$ through strong post-selection}, retaining only those trials in which all parties detect a photon. Such post-selection may allow the heralding event to depend on the measurement process and, therefore, does \emph{not} certify a genuine common cause.  

A proper test of our inequality would require an \emph{event-ready, heralded} $\mathrm{GHZ}_4$ source, where the heralding signal is independent of all measurement outputs.
Heralded GHZ generation has only been demonstrated for \emph{three} parties to date, for example in Refs.~\cite{chen2024heralded,cao2024photonic}. Extending these heralded schemes to four spacelike-separated nodes remains an open experimental challenge.

\emph{Discussion---}
In this Letter, we introduced the \emph{Coordination Principle}, a multipartite extension of Reichenbach’s Principle (or its contrapositive, Independence). It states that perfect coordination among $N$ parties is possible only if the parties share a common cause. We showed that \emph{this principle is not implied by the standard causal assumptions} of No-Signaling and Independence: there exist theories of information satisfying them in which perfect coordination can occur without any common cause.

However, we then proved that \emph{quantum theory satisfies this Coordination Principle} and derived noise-tolerant Bell-like inequalities whose violation certifies the presence of a common cause. Our inequalities are a first step towards the strongest operational constraints implied by the Coordination Principle in quantum theory. Tighter bounds can be obtained by pushing quantum inflation to higher levels; however, the computational cost grows rapidly, making these improvements prohibitive (often infeasible in practice) with current methods and hardware.
We further established that GHZ-type quantum coordination requires a quantum common cause and derived
noise-tolerant inequalities capable of witnessing such coordination in noisy experiments.

Overall, our results motivate the search for a quantitative Coordination Principle: a causal principle that not only forbids \emph{perfect} coordination without a common cause, but also quantitatively constrains how well parties lacking a shared common cause can coordinate in \emph{any} reasonable causal theory of information. Formulating such a principle would require additional conceptual input, since one must justify the numerical thresholds that should bound imperfect coordination—an issue that appears highly nontrivial. In this work, we adopted a conservative position: any quantitative version of the principle must, at minimum, imply the qualitative Coordination Principle introduced here. 

Developing such quantitative version, tight enough to yield experimentally accessible bounds while remaining theory-independent, remains an open and promising direction for future research. We hypothesize that such a quantitative formulation will be strong enough to rule out more complex forms of coordination than mere agreement on a shared random output. %

\emph{Acknowledgments---}
We thank Victor Gitton, Tein van der Lugt, Marina Maciel Ansanelli, Roberto D. Baldij\~{a}o, Fatemeh Moradi, Xiangling Xu, Peter Brown, Augustin Vanrietvelde and Renato Renner for fruitful discussions. AC would like to thank Prof. Marc-Olivier Renou and Prof. Renato Renner for supervising his master's thesis, which in part led to the present publication. %
A.C., L.T. and M.-O. R. acknowledge funding by the ANR for the JCJC grants LINKS (No. ANR-23-CE47-0003), the T-ERC QNET (No. ANR-24-ERCS-0008), the project QUANTINT, as well as the European Union’s Horizon 2020 Research and Innovation Programme under QuantERA Grant Agreements No. 731473 and No. 101017733. This work was funded by the European Union under the Marie Skłodowska-Curie Actions (MSCA) through the QNETS project (grant agreement ID: 101208259). Views and opinions expressed are, however, those of the author(s) only and do not necessarily reflect those of the European Union or the European Education and Culture Executive Agency (EACEA). Neither the European Union nor EACEA can be held responsible for them.
MCA and DC also acknowledge support from the Natural Sciences and Engineering Research Council of Canada (grants 50505-11449 and 50505-11450). Research at Perimeter Institute is supported in part by the Government of Canada through the Department of Innovation, Science and Economic Development and by the Province of Ontario through the Ministry of Colleges and Universities.

\nocite{apsrev42Control}
\bibliographystyle{apsrev4-2-wolfe}
\setlength{\bibsep}{.4\baselineskip plus .1\baselineskip minus .1\baselineskip}
\nocite{MasterThesisAntoine}
\bibliography{Refs}

\end{document}